\documentclass[preprint,11pt]{elsarticle}

\usepackage{amsmath,amssymb,amsthm,amscd}
\usepackage{graphicx,mathrsfs}

\numberwithin{equation}{section}
\newtheorem{Thm}{Theorem}[section]
\newtheorem{Lem}[Thm]{Lemma}

\newtheorem{Prop}[Thm]{Proposition}

\newtheorem{Rem}[Thm]{Remark}

\newcommand{\R}{\mathbb{R}}

\newcommand{\N}{\mathbb{N}}

\newcommand{\V}{\mathbb{V}}

\begin{document}
\begin{frontmatter}

\title{High energy solutions of the Choquard equation
} 


\author[C.D]{Daomin Cao}
\ead{dmcao@amt.ac.cn}

\author[H.L]{Hang Li\corref{cor1}}
\ead{hli@amss.ac.cn}

\cortext[cor1]{Corresponding author.}
\address[C.D]{School of Mathematics and Information Science,
Guangzhou University,
Guangzhou 510405, Guangdong, P.R.China\\
and Institute of Applied Mathematics, AMSS, Chinese Academy of Science, Beijing 100190,  P.R. China}
\address[H.L]{Institute of Applied Mathematics, Chinese Academy of Science, Beijing 100190, and University of Chinese Academy of Sciences, Beijing 100049,  P.R. China}

\begin{abstract}

In this paper we are concerned with the existence of positive high energy solutions of the Choquard equation. Under certain assumptions, the ground state of Choquard equation does not exist. However, by global compactness analysis, we prove that there exists a positive high energy solution.

\end{abstract}
\end{frontmatter}



\section{Introduction}
In this paper, we study the  following Choquard equation
\begin{equation}\label{1-0}
\begin{cases}
 -\Delta u+u=Q(x)\left(I_\alpha*\vert u\vert^p\right)\vert u\vert^{p-2}u \ \ \ in\ \ \R^N,\\
u\in H^1(\R^N),
\end{cases}
\end{equation}
where $I_\alpha(x)$ is the Riesz potential of order $\alpha \in (0,N)$ on the Euclidean space $\R^N$,
defined for each point $x \in \R^N\backslash \{0\}$ by
\begin{equation}\nonumber
I_\alpha(x)=\frac{A_\alpha}{\vert x\vert^{N-\alpha}},\ \  where\ A_\alpha=\frac{\Gamma(\frac{N-\alpha}{2})}{\Gamma(\frac{\alpha}{2})\pi^{\frac{N}{2}}2^\alpha}
\end{equation}
and $Q(x)$ is a positive bounded continuous function on $\R^N$. We consider the existence of high energy solutions under the assumptions that $\alpha=2, p=2, N=3,4,5$ or $\alpha=2,\, N=3,\ 2<p<\frac{7}{3}$.

When $\alpha=2, p=2,\, N=3$, (\ref{1-0}) is usually called the Choquard-Pekar equation which can be traced back to the 1954's work by Pekar on quantum theory of a Polaron \cite{S.P} and to 1976's model of Choquard of an electron trapped in its own hole, in an approximation to Hartree-Fock theory of one-component plasma \cite{L.E}. What's more, some Schr$\ddot{o}$dinger-Newton equations were regarded as the Choquard type equation.

When $Q(x)$ is a positive constant and $1+\frac{\alpha}{N}\leq p\leq\frac{N+\alpha}{N-2},\, N\geq 3$, the existence of positive ground state solutions of (\ref{1-0}) has been studied
in many papers,see \cite{L.E, L.P, M.L, M.G, M.V}, for instance. In addition, uniqueness of positive solutions of the Choquard equations has also been widely discussed in recent years by a lot of papers \cite{L.E1, L.E, T.K, W.T, X.C}. In \cite{W.T}, T.Wang and Taishan Yi proved that the positive solution of (\ref{1-0}) is uniquely determined, up to translation provided $\alpha=2, p=2,N=3,4,5$. The assumption on $p=2$ can be extended to $p>2$ and close to 2 when $N=3,\ \alpha=2$, and under these assumptions C.L. Xiang  proved that the positive solution of (\ref{1-0}) is unique in \cite{X.C}. What's more, in \cite{M.V},V. Moroz and J.Van Schaftingen gave some results on the decay of ground state solutions of the Choquard equation which will be used in the proof of our results.(Another result on decay of ground states was shown in \cite{X.C})

Motivated by D. Cao's work \cite{D.C,Cao}, we prove in this paper that there exists a positive high energy solution of the Choquard equation under the following condition on $Q(x)$:

(\textbf{C}): $\lim_{\vert x\vert\rightarrow +\infty} Q(x)=\bar{Q}>0,\,\,\, Q(x)\geq\frac{\sqrt 2}{2}\bar{Q},\,x\in\R^N $.\\
One can also find some other assumptions on $Q(x)$ under which similar results can be obtained. In particular we would like to mention the results in \cite{A.B} in which A. Bahri and Y.Y. Li
showed that there exists a positive solution of certain semilinear elliptic equations in $\R^N$ even if the ground state can not be achieved.

The limiting problem of (\ref{1-0}) is as following
\begin{equation}\label{1-1}
\begin{cases}
 -\Delta u+u=\bar{Q}\left(I_\alpha*\vert u\vert^p\right)\vert u\vert^{p-2}u \ \ \ in\ \ \R^N,\\
u\in H^1(\R^N),
\end{cases}
\end{equation}
where  $\bar{Q}$ is the positive constant given in condition (\textbf{C}).

When $\alpha=2, p=2, N=3$, Schr$\ddot{o}$dinger-Newton equation can be regarded as the Choquard type equation.
In \cite{G.V}, Giusi Vaira proved existence of positive bound solutions of a particular Schr$\ddot{o}$dinger-Newton type systems. However the structure of equation in \cite{G.V} is different from ours and we extend the assumption on $N$ and $p$ as well. One of the difficulties to prove our results is that the Brezis-Lieb lemma can not be applied directly to our proof. In order to overcome the difficulty we improve the results of lemma 2.2 in \cite{G.V1} for $N=3,4,5$ or $N=3,\ 2<p<\frac{7}{3}$. What's more, the main method of our proof depends on global compactness analysis and min-max method.

Our main results are as following
\begin{Thm}\label{1-3}
Assume that condition (\textbf{C}) holds, $\alpha=2, p=2, N=3,4,5$, then when the ground state level can not be achieved (\ref{1-0}) has a positive high energy solution.
\end{Thm}
\begin{Rem}\label{1-4}
Suppose that $\bar{Q}-Q(x)\geq 0$ holds in $\R^N$ and $Q(x)$ is not a constant, then it is not difficult to see that the ground state does not exist.
\end{Rem}
\begin{Rem}
When $N=3, \alpha=2$ and $2<p<\frac{7}{3}$ the uniqueness result in \cite{X.C} implies that the positive solution of (\ref{1-1}) is unique up to a translation. Moreover if we replace condition (\textbf{C}) by the following condition:\\
($\textbf{C}^{\,*}$): $\lim_{\vert x\vert\rightarrow +\infty} Q(x)=\bar{Q}>0,\ Q(x)\geq 2^{1-p}\bar{Q},\,\,\,x\in\R^N$ \\
then using the uniqueness result and condition ($\textbf{C}^{\,*}$) we can see that the result of Theorem \ref{1-3} is also true by a similar discussion without bringing about new difficulties.
\end{Rem}

 Our paper is organized as follows. In section 2 we first give some notations and preliminary results for our proof of Theorem \ref{1-3}. In section 3, we give the proof of Theorem \ref{1-3}.

\section{Some Notations and Preliminary Results}

In this section we give some preliminary results which will be used in our discussion in next section. To start with, let us first give some definition. Define
\begin{eqnarray*}
I(u)&=&\frac{1}{2}\int_{\R^N}\vert \nabla u\vert^2+u^2-\frac{1}{4}\int_{\R^N}Q(x)I_\alpha*\vert u\vert^2\vert u\vert^2,\\
I^*(u)&=&\frac{1}{2}\int_{\R^N}\vert \nabla u\vert^2+u^2-\frac{1}{4}\int_{\R^N}\bar{Q}I_\alpha*\vert u\vert^2\vert u\vert^2,\\
J(u)&=&\int_{\R^N}\vert \nabla u\vert^2+u^2,\\
\V&=&\{u\,\,|\,\,u \in H^1(\R^N),u\geq 0, \int_{\R^N}Q(x)I_\alpha*\vert u\vert^2\vert u\vert^2=1\},\\
\V^*&=&\{u\,\,|\,\,u \in H^1(\R^N),u\geq 0, \int_{\R^N}\bar{Q}I_\alpha*\vert u\vert^2\vert u\vert^2=1\}.
\end{eqnarray*}
Let $M$ and $M^*$ be defined respectively by
\begin{equation}\nonumber
M=inf\{J(u)\,\,|\,\,u \in \V\}\ \ \ \  and\ \ \ \ \ M^*=inf\{J(u)\,\,|\,\,u \in \V^*\}.
\end{equation}
$D^{1,2}(\R^N)$ is the completion of $C^\infty_0(\R^3)$ with respect to the norm
\begin{equation}
\|u\|^2_{D^{1,2}}=\int_{\R^3}\vert \nabla u\vert^2dx.
\end{equation}

For later discussion, we introduce an inequality given in \cite{H.G}.
\begin{Prop}(Hardy-Littlewood-Sobolev inequality \cite{H.G})\label{2-0}
Let $q \in (1,+\infty)$ and  $\alpha< \frac{N}{q}$, then for every $f\in L^q(\R^N)$,
 $I_\alpha*f \in L^\frac{Nq}{N-\alpha q}(\R^N)$ and
\begin{equation*}
(\int_{\R^N}\vert I_\alpha*f\vert^\frac{Nq}{N-\alpha q})^\frac{N-\alpha q}{Nq}\leq C_{N,\alpha,q}(\int_{\R^N}\vert f\vert^q)^\frac{1}{q}.
\end{equation*}
\end{Prop}
By Proposition \ref{2-0}, we have
\begin{eqnarray}\label{2-1}
\int_{\R^N}(I_\alpha*\vert u\vert^p)\vert u\vert^p&\leq& C_{N,\alpha}(\int_{\R^N}\vert u\vert^\frac{2Np}{N+\alpha})^{1+\frac{\alpha}{N}}\nonumber\\
&\leq& C(\int_{\R^N} \vert\nabla u\vert^2+\vert u\vert^2)^p.
\end{eqnarray}
As a consequence of (\ref{2-0}) we can easily get
\begin{Prop}\label{2-2}
Suppose $1+\frac{\alpha}{N}\leq p< \frac{N+\alpha}{N-2}$ and $\alpha \in (0,N)$. If $u_m\rightharpoonup 0$ in $H^1(\R^N)$, then for any bounded domain $\Omega$ in $\R^N$,
\begin{equation*}
\int_\Omega (I_\alpha*\vert u_m\vert^p)\vert u_m\vert^p\rightarrow 0.
\end{equation*}
\end{Prop}

Let $u$ be a positive ground state solution of (\ref{1-1}),  following \cite{M.V, X.C} the decay of $u$ is as follows
\begin{Prop}\label{1-2}
Assume that $\alpha=2, p=2,N=3,4,5$ or $\alpha=2,\ N=3,\ 2<p<\frac{7}{3}$, then
$u=O(e^{-\sigma\vert x\vert})$ for $\vert x\vert$ large enough,
where $\sigma$ is a positive constant.
\end{Prop}

Next we shall prove a proposition on the weak convergence of a nonlinear operator. Denote
\begin{equation}
T(u,v,w,z)=\int_{\R^N}\int_{\R^N}\frac{u(x)v(x)w(y)z(y)}{\vert x-y\vert^{N-2}}dxdy.
\end{equation}
\begin{Prop}\label{2-9}
Assume that $3\leq N\leq 6$ and that there are three weekly convergent sequences in $H^1(\R^N)$ such that $u_m\rightharpoonup u,\ v_m\rightharpoonup v,\ w_m\rightharpoonup w$ and $z \in\ H^1(\R^N)$, then as $m\rightarrow +\infty$
\begin{equation*}
T(u_m,v_m,w_m,z)\rightarrow T(u,v,w,z).
\end{equation*}
\end{Prop}
\begin{proof}
Firstly assume that $u_m\equiv u$ for all $m$, we claim that $T(u,v_m,w_m,z)\rightarrow T(u,v,w,z)$.
\begin{equation*}
T(u,v_m,w_m,z)=T(u,v_m-v,w_m,z)+T(u,v,w_m,z).
\end{equation*}
Since $w_m\rightharpoonup w$ in $H^1(\R^N)$, then $w_m\rightharpoonup w$ in both $L^2(\R^N)$ and $L^\frac{2N}{N-2}(\R^N)$.
When $N=3,4$, since
\begin{eqnarray*}
&&\int_{\R^N}(\int_{\R^N}\frac{u(x)v(x)}{\vert x-y\vert^{N-2}}dxz(y))^2dy\\
&\leq&(\int_{\R^N}(\int_{\R^N}\frac{u(x)v(x)}{\vert x-y\vert^{N-2}}dx)^\frac{2N}{N-2})^{\frac{N-2}{N}}(\int_{\R^N}\vert z\vert^N)^{\frac{2}{N}},
\end{eqnarray*}
which implies that $\int_{\R^N}\frac{u(x)v(x)}{\vert x-y\vert^{N-2}}dxz(y)\in\ L^2(\R^N)$. Therefore it is easy to prove that $T(u,v,w_m,z)\rightarrow T(u,v,w,z)$.

For $5\leq N\leq 6$, similarly we have
\begin{eqnarray*}
&&\int_{\R^N}(\int_{\R^N}\frac{u(x)v(x)}{\vert x-y\vert^{N-2}}dxz(y))^{\frac{2N}{N+2}}dy\\
&\leq&(\int_{\R^N}(\int_{\R^N}\frac{u(x)v(x)}{\vert x-y\vert^{N-2}}dx)^{\frac{N}{2}})^{\frac{4}{N+2}}(\int_{\R^N}\vert z\vert^\frac{2N}{N-2})^{\frac{N-2}{N+2}}.
\end{eqnarray*}
As a consequence, $\int_{\R^N}\frac{u(x)v(x)}{\vert x-y\vert^{N-2}}dxz(y)\in\ L^{\frac{2N}{N+2}}$ from which we get
$T(u,v,w_m,z)\rightarrow T(u,v,w,z)$.

In addition, using Holder inequality we have
\begin{eqnarray*}
&&T(u,v_m-v,w_m,z)^2\\
&\leq&\int_{\R^N}\int_{\R^N}\frac{(v_m-v)^2(x)z^2(y)}{\vert x-y\vert^{N-2}}dxdy\int_{\R^N}\int_{\R^N}\frac{u^2(x)w_m^2(y)}{\vert x-y\vert^{N-2}}dxdy\\
&=&T(v_m-v,v_m-v,z,z)T(u,u,w_m,w_m).
\end{eqnarray*}
It is easy to see that $T(u,u,w_m,w_m)$ is bounded.

Set $\phi_{u^2}(y)=\int_{\R^N}\frac{u^2(x)}{\vert x-y\vert^{N-2}}dx$, then $\phi_{u^2} \in D^{1,2}(\R^N)$ is a solution of
\begin{equation*}
-\Delta\phi=u^2  \ \ \ in\ \R^N
\end{equation*}
and we have, as $m\rightarrow +\infty$,
\begin{equation*}
T(v_m-v,v_m-v,z,z)=\int_{\R^N}\phi_{z^2}(v_m-v)^2dx\rightarrow 0.
\end{equation*}
Thus we complete the claim.
Now consider that
\begin{eqnarray*}
T(u_m,v_m,w_m,z)=T(u,v_m,w_m,z)+T(u_m-u,v_m,w_m,z)
\end{eqnarray*}
$T(u,v_m,w_m,z)\rightarrow T(u,v,w,z)$ and with respect to the above discussion we get that $T(u_m-u,v_m,w_m,z)\rightarrow 0$ as $m\rightarrow +\infty$.
\end{proof}

\begin{Lem}\label{2-8}
Assume that $3\leq N\leq 6,\ \alpha=2$ and that $\{u_m\}$ is bounded in $H^1(\R^N)$. If $u_m\rightarrow u$ almost everywhere on $\R^N$ as $m\rightarrow +\infty$, then
\begin{eqnarray*}
T(u_m,u_m,u_m,u_m)-T(u,u,u,u)=T(u_m-u,u_m-u,u_m-u,u_m-u)+o(1).
\end{eqnarray*}
\end{Lem}
\begin{proof}
\begin{eqnarray*}
&&T(u_m,u_m,u_m,u_m)\\
&=&T(u_m,u_m,u_m,u_m-u)+T(u_m,u_m,u_m,u),\\
&=&T(u_m,u_m,u_m,u_m-u)+T(u,u,u,u)+o(1),\\
&=&T(u_m-u,u_m-u,u_m-u,u_m-u)+T(u,u,u,u)+o(1).
\end{eqnarray*}
\end{proof}
\begin{Rem}
For $N=3, \alpha=2$ and $2<p<\frac{7}{3}$ the results of Proposition \ref{2-9} and Lemma \ref{2-8} are also true by a similar calculation.
\end{Rem}

Next, we establish a global compactness lemma.
\begin{Lem}\label{2-3}
Let $\{u_m\} \subset H^1(\R^N)$ be a sequence such that as $m\rightarrow +\infty$
\begin{enumerate}
\item[(i)] $I(u_m)\rightarrow C$ ,\\
\item[(ii)] $I^{'}(u_m)\rightarrow 0$ \ \ \  in\ $H^{-1}(\R^N)$.
\end{enumerate}
Then, there exists a number $k \in \N$, $k$ sequences of points $\{y^j_m\}$ such that $\vert y^j_m\vert\rightarrow +\infty$ as $m\rightarrow +\infty$, $1\leq j\leq k$, $k+1$ sequence of functions $\{u^j_m\}\subset H^1(\R^N)$, $0\leq j\leq k$, such that for some subsequences
\begin{equation*}
\begin{cases}
u^0_m\equiv u_m\rightharpoonup u^0,\\
u^j_m=(u^{j-1}_m-u^{j-1})(x-y^j_m)\rightharpoonup u^j,\\
1\leq j\leq k.
\end{cases}
\end{equation*}
where $u^0$ is a solution of (\ref{1-0}) and $u^j, 1\leq j\leq k$ are nontrivial positive solutions of (\ref{1-1}). Moreover as $m\rightarrow +\infty$
\begin{eqnarray*}
J(u_m)&\rightarrow& \sum^k_{j=0} J(u^j),\\
I(u_m)&\rightarrow& I(u^0)+\sum^k_{j=1}I^*(u^j).
\end{eqnarray*}
\end{Lem}
\begin{proof}
Our proof is similar to the those in \cite{B.V} and \cite{Z.X}.
Since $\{u_m\}$ is $(PS)_C$ sequence of $I(u)$, it is easy to prove that $u_m$ is bounded in $H^1(\R^N)$. Then we can assume that
$u_m\rightharpoonup u^0$ in $H^1(\R^N)$. Set $v_m=u_m-u^0$, then  $v_m\rightharpoonup 0$ in $H^1(\R^N)$.
If $v_m\rightarrow 0$ in $H^1(\R^N)$, we are done. Now suppose that $v_m\nrightarrow 0$ in $H^1(\R^N)$.
By Proposition \ref{2-2} and Lemma \ref{2-8}, we get
\begin{eqnarray*}
I(v_m)&=&I^*(v_m)+o(1),\\
I'(v_m)&=&(I^*)'(v_m)+o(1)=o(1).
\end{eqnarray*}
Moreover there exists $\lambda \in(0,+\infty)$ such that $I^*(v_m)\geq \lambda>0$ for $m$ large enough.
In fact, otherwise $I^*(v_m)=o(1), \ \ (I^*)'(v_m)=o(1)$ would imply $\|v_m\|_{H^1}\rightarrow 0$, which is a contradiction to $v_m\nrightarrow 0$ in $H^1(\R^N)$.
Let us decompose $\R^N$ into N-dim hypercubes $\Omega_i$ and define
\begin{equation}
d_m=\sup_{\Omega_i}(\int_{\Omega_i}\bar{Q}I_\alpha*\vert v_m\vert^2\vert v_m\vert^2)^\frac{1}{4}.
\end{equation}
Claim $d_m\geq \gamma>0$. Since $(I^*)'(v_m)=o(1)$ as $m\rightarrow +\infty$, then
\begin{eqnarray*}
\|v_m\|_{H^1}&=&\int_{\R^N}\bar{Q}I_\alpha*\vert v_m\vert^2\vert v_m\vert^2+o(1),\\
I^*(v_m)&=&\frac{1}{4}\int_{\R^N}\bar{Q}I_\alpha*\vert v_m\vert^2\vert v_m\vert^2+o(1).
\end{eqnarray*}
Thus, we have
\begin{eqnarray*}
4I^*(v_m)+o(1)&=&\int_{\R^N}\bar{Q}I_\alpha*\vert v_m\vert^2\vert v_m\vert^2\\
&=&\sum_i\int_{\Omega_i}\bar{Q}I_\alpha*\vert v_m\vert^2\vert v_m\vert^2\\
&\leq& d^2_m\sum_i(\int_{\Omega_i}\bar{Q}I_\alpha*\vert v_m\vert^2\vert v_m\vert^2)^\frac{1}{2}\\
&\leq& C_Nd^2_m\sum_i\|v_m\|^2_{H^1(\Omega_i)}\ \ (by\ (\ref{2-1}))\\
&=&C_Nd^2_m\|v_m\|^2_{H^1(\R^N)},
\end{eqnarray*}
where $C_N$ is a positive constant.
Since $I^*(v_m)\geq \lambda>0$ then $d_m\geq \gamma>0$.
Now, let us call $y_m$ the center of $\Omega_m$ such that
\begin{equation*}
(\int_{\Omega_m}\bar{Q}I_\alpha*\vert v_m\vert^2\vert v_m\vert^2)^\frac{1}{4}\geq d_m-\frac{1}{m}
\end{equation*}
and put $\widetilde{v_m}=v_m(x+y_m)$. It is easy to prove that $\widetilde{v_m}\rightharpoonup v_0\not\equiv 0$. In fact, letting $\Omega$ be the hypercube centered at the origin,then we have
\begin{equation}\label{2-4}
(\int_{\Omega}\bar{Q}I_\alpha*\vert \widetilde{v_m}\vert^2\vert \widetilde{v_m}\vert^2)^\frac{1}{4}=(\int_{\Omega_m}\bar{Q}I_\alpha*\vert v_m\vert^2\vert v_m\vert^2)^\frac{1}{4}\geq d_m-\frac{1}{m}\geq \gamma+o(1)
\end{equation}
If $\widetilde{v_m}\rightharpoonup 0$, then $\int_{\Omega}\bar{Q}I_\alpha*\vert \widetilde{v_m}\vert^2\vert \widetilde{v_m}\vert^2\rightarrow 0$ as $m\rightarrow +\infty$, we get a contradiction. Iterating the above procedure, if $\widetilde{v_m}\rightarrow v_0$ we are done, otherwise setting $w_m=\widetilde{v_m}-v_0\rightharpoonup 0$ and $w_m\nrightarrow 0$, continue the above procedure. Since $\{u_m\}$ is bounded away from zero, by Brezis-Lieb lemma and Lemma \ref{2-8} we know that the iteration must terminate at some index $k>0$ and
\begin{eqnarray*}
I(u_m)&=&I(u^0)+\sum^k_{j=1}I^*(u^j)+o(1),\\
J(u_m)&=&\sum^k_{j=0}J(u^j)+o(1).
\end{eqnarray*}
Moreover we claim as $m\rightarrow +\infty$, $\vert y_m\vert\rightarrow +\infty$, otherwise $\vert y_m\vert$ is bounded, we can choose a bounded domain $\Sigma$ such that  $\bigcup\Omega_m\subset \Sigma$.
As a consequence, by $v_m\rightharpoonup 0$, we get $\int_{\Sigma}\bar{Q}I_\alpha*\vert v_m\vert^2\vert v_m\vert^2\rightarrow 0$, which is a contradiction to (\ref{2-4}).
\end{proof}
\begin{Rem}\label{2-5}
For $1\leq j\leq k$, $J(u^j)\geq {M^*}^2$ and $I^*(u^j)\geq \frac{1}{4}{M^*}^2$.
If $c \in (0,\frac{1}{4}{M^*}^2)$, then we can see that $k=0$ and therefore
\begin{equation*}
u_m\rightarrow u^0\not\equiv 0.
\end{equation*}
If $c \in [\frac{1}{4}{M^*}^2,\frac{1}{2}{M^*}^2)$, then either $k=0$ or $k=1$.
\begin{equation*}
\begin{cases}
u_m(x)\rightarrow u^0(x)\ \ \ \ k=0,\\
u_m(x)=u^0(x)+u(x+y_m)+w_m(x)\ \ \ k=1,
\end{cases}
\end{equation*}
where $w_n(x)\rightarrow 0$ in $H^1(\R^N)$.
\end{Rem}
\begin{Lem}\label{2-6}
Assume that $\{u_m\}$ is a $(PS)_c$ sequence of $I(u)$ and $M=M^*$. If $0<c<\frac{1}{4}{M^*}^2$ or $\frac{1}{4}{M^*}^2<c<\frac{1}{2}{M^*}^2$,
then $\{u_m\}$ contains a strongly convergent subsequence.
\end{Lem}
\begin{proof}
If $0<c<\frac{1}{4}{M^*}^2$ we are done. If $\frac{1}{4}{M^*}^2<c<\frac{1}{2}{M^*}^2$, since $M=M^*$, we get $J(u^0)\geq \frac{1}{4}{M^*}^2$, then $u^0\equiv 0$ or $u\equiv 0$.
If $u\equiv 0$, we are done. Otherwise, $u^0\equiv 0$ and $u\not\equiv 0$ is a positive solution of (\ref{1-1}). By the uniqueness of positive solutions of (\ref{1-1}), $I^*(u)=\frac{1}{4}{M^*}^2$ which contradicts to the value of $c$.
\end{proof}

Another form of Lemma \ref{2-6} is as following
\begin{Lem}\label{2-7}
Assume that $\{u_m\}\subset \V$ such that
\begin{equation}
\begin{cases}
(i) \ \ J(u_m)\rightarrow c \ \ \ \in (0,M^*) \ or \ c \ \in (M^*,\sqrt{2}M^*),\\
(ii)\ \ dJ|_{\V}(u_m)\rightarrow 0.
\end{cases}
\end{equation}
then, $J|_{\V}$ has a critical point $v_0$ such that $J(v_0)=c$.
\end{Lem}

\section{Proof of Theorem \ref{1-3}}
It is easy to see that $0<M\leq M^*$, from the fact that $Q(x)\rightarrow \bar{Q}$ as $\vert x\vert\rightarrow +\infty$. Under the assumptions in Remark \ref{1-4}, when the ground state is achieved it is easy to see that $ M^*< M$ which is a contradiction.

If $M<M^*$, there must exist a sequence $\{u_m\}\subset \V$ such that as $m\rightarrow +\infty$
\begin{equation}
J(u_m)\rightarrow M\ \ \ \ dJ|_{\V}(u_m)\rightarrow  0.
\end{equation}
Consequently, by Lemma \ref{2-7}, $J|_{\V}$ has a critical point $v_0 \in \V$ such that $J(v_0)=M,\ \ \  dJ|_{\V}(v_0)=0$.
Taking $u^0=M^\frac{1}{2}v_0$, it is easy to see that $u^0$ is a positive solution of (\ref{1-0}) and $I(u^0)=\frac{1}{4}M^2$.
If $M=M^*$ and M can be achieved in $\V$, there also a positive solution of (\ref{1-0}). Next we always assume $M=M^*$ and $M$ can not be achieved.
Defined $\beta(u):\ H^1(\R^N)\rightarrow \R^N$ as following
\begin{equation}
\beta(u)=\int_{\R^N}u^2\chi(\vert x\vert)\cdot x,
\end{equation}
where
\begin{eqnarray}\chi(t)=
\begin{cases}
1\ \ \ \ \ \ \  0\leq t\leq 1\\
\frac{1}{t}\ \ \ \ \ \ t>1.
\end{cases}
\end{eqnarray}
Let $\overline{\V}$ be defined as $\overline{\V}=\{u|u\in \V, \beta(u)=0\}$ and $\bar{u}$ be a positive solution of (\ref{1-1}) achieving its maximum at the origin.
\begin{Lem}\label{3-0}
Let $\bar{M}=\inf\{J(u)\,|\,u\in \overline{\V}\}$. If $M=M^*$ can not be achieved in $\V$, then $M<\bar{M}$ and there exists $R>0$ such that
\begin{eqnarray*}
\begin{cases}
(i)\ \ \ J(h(y))\ \in \ \ (M,\frac{M+\bar{M}}{2})\ \ \ if \ \vert y\vert\geq R,\\
(ii)\ \ \ (\beta\cdot h(y),y)>0 \ \ \ \ \ \ \ \ \ \ \ if\ \ \vert y\vert=R,
\end{cases}
\end{eqnarray*}
where $h(y)=\bar{u}(x-y)/(\int_{\R^N}Q(x)(I_2*\vert \bar{u}(x-y)\vert^2)\vert \bar{u}(x-y)\vert^2)^\frac{1}{4}$.
\end{Lem}
\begin{proof}
It is obvious that $\bar{M}\geq M=M^*$. To prove $\bar{M}>M$, we shall argue by contradiction. Suppose $\bar{M}=M$, then there exists a sequence $\{u_m\}\subset \overline{\V}$ such that $J(u_m)\rightarrow M$ and $dJ|_{V}(u_m)\rightarrow 0$ as $m\rightarrow +\infty$. There exists $u^0$ such that $u_m\rightharpoonup u^0$ in $H^1(\R^N)$.
Let $v_m=M^\frac{1}{2}u_m$, we deduced that as $m\rightarrow +\infty$
\begin{eqnarray}
I(v_m)&\rightarrow&\frac{1}{4}M^2,\\
dI(v_m)&\rightarrow& 0.
\end{eqnarray}
Denote $v^0=M^\frac{1}{2}u^0$ from Remark \ref{2-5}, we have
\begin{equation}
v_m(x)=v^0+u(x-y_m)+w_m(x),
\end{equation}
where $u$ is either 0 or positive solution of (\ref{1-1}) and $w_m\rightarrow 0$.
If $u\equiv 0$, we get $v^0\not\equiv 0,\ v_m\rightarrow v^0$. Thus $u_m\rightarrow u^0$ and $u^0 \in \V,\ J(u^0)=M$ which is a contradiction.
So $u\not\equiv 0$, hence $v^0\equiv 0$.

Let us set $(\R^N)^{+}_m=\{x\in \R^N : (x,y_m)>0\}$ and $(\R^N)^{-}_m=\R^N\backslash (\R^N)^{+}_m$. Choosing $m$ large enough, since $\vert y_m\vert\rightarrow +\infty$, we can assert that there is a ball $B_r(y_m)=\{x \in \R^N : \vert x-y_m\vert<r\} \subset (\R^N)^{+}_m$ such that $\forall\ x\in B_r(y_m),\ \ u(x-y_m)\geq\frac{1}{2}u(0)>0$.
By Proposition \ref{1-2}
\begin{eqnarray}\label{3-1}
&&(\beta(u(x-y_m)),y_m)\nonumber\\
&=&\int_{(\R^N)^{+}_m}u(x-y_m)\chi(\vert x\vert)(x,y_m)+\int_{(\R^N)^{-}_m}u(x-y_m)\chi(\vert x\vert)(x,y_m)\nonumber\\
&\geq&\int_{B_r(y_m)}\frac{1}{2}u(0)\chi(\vert x\vert)(x,y_m)-\int_{(\R^N)^{-}_m}\frac{kR\vert y_m\vert}{e^{\sigma\vert x-y_m\vert}}\nonumber\\
&\geq& C-o(\frac{1}{\vert y_m\vert}).
\end{eqnarray}
where $C$ is a positive constant.
Thus $\beta(v_m)\neq 0$ for $m$ large enough. So $\beta(u_m)\neq 0$ for large $m$ which is a contradiction.

By $Q(x)\rightarrow \bar{ Q}$ as $\vert x\vert\rightarrow +\infty$, it is easy to check that $h(y)$ is continuous on $y$ and $J(h(y))\rightarrow M^*=M$, then (i) is satisfied by choosing $R>0$ large enough. (ii) is analogous to the calculation of (\ref{3-1}), $(\beta(h(y)),y)>0$ if $\vert y\vert=R$.
\end{proof}

For fixed $R$ define
\begin{eqnarray}
F&=&\{f\in C(\overline{B_R},\V):f|_{\partial B_R}=h|_{\partial B_R}\},\\
c&=&\inf_{f\in F}\max_{y\in \overline{B_R}}J(f(y)).
\end{eqnarray}

Since Lemma \ref{3-0} (ii), by Brouwer degree, for any $f\in F$ there exists a point $y\in B_R$ such that $\beta(f(y))=0$ and consequently $f(y)\in \overline{\V}$. So $c\geq \bar{M}>M=M^*$. Condition (\textbf{C}) deduces for $y\in \R^N$
\begin{equation*}
\int_{\R^N}Q(x)(I_2*\vert \bar{u}(x-y)\vert^2)\vert \bar{u}(x-y)\vert^2>\frac{\sqrt 2}{2}\int_{\R^N}\bar{Q}(I_2*\vert \bar{u}(x-y)\vert^2)\vert \bar{u}(x-y)\vert^2
\end{equation*}
As a consequence, we get
\begin{eqnarray}
\max_{y\in \overline{B_R}}J(h(y))&=&\max_{y\in \overline{B_R}}\frac{\int_{\R^N}\vert \nabla \bar{u}\vert^2+\vert \bar{u}\vert^2}{(\int_{\R^N}Q(x)(I_2*\vert \bar{u}(x-y)\vert^2)\vert \bar{u}(x-y)\vert^2)^\frac{1}{2}}\nonumber\\
&<&\sqrt 2\frac{\int_{\R^N}\vert \nabla \bar{u}\vert^2+\vert \bar{u}\vert^2}{(\int_{\R^N}\bar{Q}(I_2*\vert \bar{u}(x)\vert^2)\vert \bar{u}(x)\vert^2)^\frac{1}{2}}\nonumber\\
&=&\sqrt 2M^*.
\end{eqnarray}
$M^*<\bar{M}\leq c<\sqrt 2M^*$. By Lemma \ref{3-0} (i)
\begin{equation}
\max_{y\in \partial B_R}J(h(y))<\frac{M+\bar{M}}{2}<\bar{M}<c.
\end{equation}
Thus by Lemma \ref{2-7}, we conclude that $J|_{\V}$ has a critical point $v_0$ such that $J(v_0)=c$, $dJ|_{\V}(v_0)=0$.
Let $u^0=c^\frac{1}{2}v_0$, then it is easy to see that $u^0$ is a positive high energy solution of (\ref{1-0}) and $I(u^0)=\frac{1}{4}c^2<\frac{1}{2}{M^*}^2$.
Thus we complete the proof of Theorem \ref{1-3}.

\vspace{0.5cm}
\noindent{\bf Acknowledgments:}
 This work was partially supported by NSFC grants (No.11771469 and No.11688101). Cao was also supported by the Key Laboratory of Random Complex Structures and Data Science, AMSS, Chinese Academy of Sciences (2008DP173182).

\newpage
{\bf References}


\end{document}